\documentclass{article}

\usepackage{arxiv}

\renewcommand{\headeright}{}
\renewcommand{\undertitle}{}
\renewcommand{\shorttitle}{Near-Optimal Min-Max Optimization Made Simple}
\usepackage[utf8]{inputenc} 
\usepackage[T1]{fontenc}    
\usepackage{hyperref}       
\usepackage{url}            
\usepackage{booktabs}       
\usepackage{amsfonts}       
\usepackage{nicefrac}       
\usepackage{microtype}      
\usepackage{color}
\usepackage{algpseudocode}
\usepackage{algorithm}
\usepackage{comment}

\usepackage{amsmath}
\usepackage{amssymb}
\usepackage{mathtools}
\usepackage{amsthm}
\usepackage{xcolor}
\hypersetup{
    colorlinks=true,
    linkcolor=black,
    urlcolor=black,
    citecolor=black
}
\usepackage[capitalize,noabbrev]{cleveref}
\newcommand{\email}[1]{\href{mailto:#1}{\texttt{#1}}}
\theoremstyle{plain}
\newtheorem{theorem}{Theorem}[section]
\newtheorem{proposition}[theorem]{Proposition}
\newtheorem{lemma}[theorem]{Lemma}
\newtheorem{corollary}[theorem]{Corollary}
\theoremstyle{definition}
\newtheorem{definition}[theorem]{Definition}
\newtheorem{assumption}[theorem]{Assumption}
\theoremstyle{remark}
\newtheorem{remark}[theorem]{Remark}

\DeclarePairedDelimiter{\norm}{\lVert}{\rVert} 

\title{Min-Max Optimization Made Simple: 
Approximating the Proximal Point Method via Contraction Maps
}

\makeatletter

\def\blfootnote{\xdef\@thefnmark{}\@footnotetext}
\newtheorem*{rep@theorem}{\rep@title}
\newcommand{\newreptheorem}[2]{%
\newenvironment{rep#1}[1]{%
 \def\rep@title{#2 \ref{##1}}%
 \begin{rep@theorem}}%
 {\end{rep@theorem}}}

\def\blfootnote{\xdef\@thefnmark{}\@footnotetext}
\makeatother

\author{%
  Volkan Cevher\\  EPFL\\
  \email{volkan.cevher@epfl.ch}\\
  \And
  Georgios Piliouras\\  SUTD\\
  \email{georgios@sutd.edu.sg}\\
  \AND
  Ryann Sim\\  SUTD\\
  \email{ryann\_sim@mymail.sutd.edu.sg}\\
  \And
  Stratis Skoulakis\\
  EPFL\\
  \email{efstratios.skoulakis@epfl.ch}
  }
\date{}

\newreptheorem{theorem}{Theorem}

\begin{document}

\maketitle


\begin{abstract}
In this paper we present a first-order method that admits near-optimal convergence rates for convex/concave min-max problems while requiring a simple and intuitive analysis. Similarly to the seminal work of Nemirovski \cite{N04} and  the recent approach of Piliouras et al.~\cite{PSS21} in normal form games, our work is based on the fact that the update rule of the \textit{Proximal Point method} (PP) can be approximated up to accuracy $\epsilon$ with only $\mathcal{O}(\log 1/\epsilon)$ additional gradient-calls through the iterations of a contraction map. Then combining the analysis of (PP) method with an error-propagation analysis we establish that the resulting first order method, called \textit{Clairvoyant Extra Gradient}, admits near-optimal time-average convergence for general domains and last-iterate convergence in the unconstrained case.  
\end{abstract}
\section{Introduction}
Given a function $f: \mathcal{D}_1 \times \mathcal{D}_2 \mapsto \mathbb R$ we consider the min-max optimization problem
\begin{equation}\label{eq:minmax}
\min_{x \in \mathcal{D}_1}\max_{y \in \mathcal{D}_2} f(x,y)
\end{equation}
where $f(x,y)$ is convex/concave, $f(\cdot, y)$ is convex for all $y \in \mathcal{D}_2$ while $f(x, \cdot)$ is concave for all $x \in \mathcal{D}_1$ and $\mathcal{D}_1$,$\mathcal{D}_2$ are considered to be convex sets. In its approximate version, Problem~\ref{eq:minmax} asks for an $\epsilon$-approximate min-max pair of points $(x^\ast,y^\ast) \in \mathcal{D}_1 \times \mathcal{D}_2$ such that 
\begin{equation}\label{eq:saddle}
f(x^\ast,y) - \epsilon \leq f(x^\ast, y^\ast) \leq f(x,y^\ast) + \epsilon~~~\text{for all }x\in \mathcal{D}_1, y \in \mathcal{D}_2
\end{equation}
where $\epsilon >0$ is an accuracy parameter. Problem~\ref{eq:saddle} admits numerous applications in Game Theory and Convex Optimization, while more recent applications include training of adversarial neural networks \cite{G14,A17,DISZ18}, robust optimization \cite{BGN09} and adversarial examples \cite{MMSTV18}.

First-order methods that only assume black-box access to the gradients $\nabla_x f(x,y)$ and $\nabla_y f(x,y)$ have been the standard choice for tackling Problem~\ref{eq:minmax}. The \textit{Proximal Point method} (PP)\footnote{The proximal point method is a fundamental method in convex optimization which solves the minimization of the cost function $f$ by iteratively solving the subproblem: $x_{k+1} = \min_u \big[f(u)+ \frac{1}{2\gamma} \lVert u-x_k\rVert^2\big]$ with step-size $\gamma$ and where the norm is the $\ell_2$ norm. 
In our setting this is equivalent to (\refeq{eq:PP}) and we will use this formulation for the rest of the paper.} has been one of the most intuitive and influential first-order methods for min-max optimization~\cite{moreau1965proximite,M70,R76}. More precisely, given a pair of points $(x_t,y_t)$ (PP) selects a new pair of points $(x_{t+1},y_{t+1})$ such that
\begin{equation}\label{eq:PP}
x_{t+1} = \left[x_t - \gamma \nabla_x f(x_{t+1}, y_{t+1} ) \right]_{\mathcal{D}_1}~~~~~ y_{t+1} = \left[y_t + \gamma \nabla_y f(x_{t+1}, y_{t+1} ) \right]_{\mathcal{D}_2}.  
\tag{PP}
\end{equation}

Unfortunately \eqref{eq:PP} is an \textit{implicit method} in the sense that its update rule requires the solution of a fixed-point problem that in its full generality cannot be efficiently computed \footnote{Notice that Equation~\eqref{eq:PP} contains $x^{t+1},y^{t+1}$ in both left and right hand sides.}. On the other hand, \eqref{eq:PP} admits optimal convergence rates\footnote{under the black-box access model \cite{R76}.} that can be established through simple and intuitive arguments. As a result, \eqref{eq:PP} has served as a very useful road-map for the design and the analysis of implementable first-order methods \cite{N04,MOP20,MS10,GPDO20}. It is known that \eqref{eq:PP} admits $\Theta(1/T)$ time-average convergence ($(\hat{x},\hat{y}):= \sum_{t=0}^{T-1}(x_t,y_t)/T$) while Golowich et al. \cite{GPDO20} recently established $\Theta(1/\sqrt{T})$ last-iterate convergence in the unconstrained case with a very intuitive argument.

The long-standing line of research on first-order methods for min-max optimization has led to implementable methods that achieve the convergence rates of \eqref{eq:PP}. Two notable examples include the \textit{Extra-Gradient method} (EG) initially proposed by Korpelevich \cite{K76} and the \textit{Optimistic Gradient Descent/Ascent method} (OGDA) initially proposed by Popov \cite{P80}. For example, given a pair of points $(x_t,y_t)$, \eqref{eq:EG} updates $(x_{t+1},y_{t+1})$ such that
\begin{equation*}\tag{EG}
\label{eq:EG}
\begin{gathered}
      x_{t+1/2} = \left[x_t - \gamma \nabla_x f(x_{t}, y_{t} ) \right]_{\mathcal{D}_1}~~~~~ y_{t+1/2} = \left[y_t + \gamma \nabla_y f(x_{t}, y_{t} ) \right]_{\mathcal{D}_2}   \\
      x_{t+1} = \left[x_t - \gamma \nabla_x f(x_{t+1/2}, y_{t+1/2} ) \right]_{\mathcal{D}_1}~~~~~ y_{t+1} = \left[y_t + \gamma \nabla_y f(x_{t+1/2}, y_{t+1/2} ) \right]_{\mathcal{D}_2} 
    \end{gathered}
\end{equation*}
Both (EG) and (OGDA) have been extensively studied and it is known that both of them admit $O(1/T)$ time-average convergence and $O(1/\sqrt{T})$ last-iterate convergence~\cite{N04,RS13,MOP20,GPDO20,COZ22,MS10}. Moreover both (EG) and (OGDA) can be interpreted as approximations of \eqref{eq:PP}, something that has been highlighted in many works~\cite{N04,RS13,MOP20,GPDO20,COZ22,MS10}. However, this approximation remains mostly an intuitive connection since the analysis of both (EG) and (OGDA) does not go directly through (PP) by explicitly handling the error term.
\subsubsection*{Our Contributions and Techniques}
We present a first-order method which we call \textit{Clairvoyant Extra-Gradient} (CEG) that comes as a very straightforward approximation of the \textit{Proximal Point method} and admits a simpler and conciser analysis  than the respective analysis of (EG) and (OGDA).

Similarly to the seminal work of Nemirovski \cite{N04} and  the recent approach of Piliouras et al.~\cite{PSS21} in normal form games, the cornerstone idea behind (CEG) is that for sufficiently small $\gamma$, the update rule of PP can be approximated within  $\Theta\left(\mathrm{poly}(\epsilon)\right)$ accuracy with only $\Theta(\log 1/\epsilon)$ extra gradient-calls through the iterations of a contraction map\footnote{$\gamma < 1/L$ where $L$ is the smoothness constant of $f(x,y)$.}. More precisely, given a pair of points $(x_t,y_t)$ (CEG) uses $\Theta(\log T)$ additional gradient-calls to compute a pair of points $(x_{t+1},y
_{t+1})$ such that
\begin{equation*}
\label{eq:CEG}
\norm{x_{t+1} - \left[x_t - \gamma \nabla_x f(x_{t+1}, y_{t+1} ) \right]_{\mathcal{D}_1}}\leq \frac{1}{\mathrm{poly}(T)}~,~~~\norm{y_{t+1} - \left[y_t + \gamma \nabla_y f(x_{t+1}, y_{t+1} ) \right]_{\mathcal{D}_2}}\leq \frac{1}{\mathrm{poly}(T)}.
\end{equation*}
Once the latter is established, the analysis of Proximal Point can be easily modified to show that Clairvoyant Extra-Gradient admits $\Theta(\log T /T)$ time-average and $\Theta(\log T /\sqrt{T})$ last-iterate convergence\footnote{the $\log T$ extra terms are due to the intermediate additional gradient-calls required by the iterations of the contraction map.}.


\begin{remark}
The idea of approximating the Proximal Point method through the iterations of a contraction map appears in a recent work in the context of online learning in games \cite{PSS21}. Piliouras et al. \cite{PSS21} proposed the \textit{Clairvoyant MWU dynamics}, which are able to \textit{predict}
in certain time-steps the joint behavior of the selfish agents by associating the update rule (that the agents independently run) with the iterations of a contraction map. As the authors show, this results in a $\log T$-sparse subsequence of play in which all agents admit $\mathcal{O}(1)$ regret which implies a $\mathcal{O}(\log T / T)$ convergence rate to \textit{Coarse Correlated Equilibrium}. The name \textit{Clairvoyant Extra Gradient} comes from the name \textit{Clairvoyant MWU} can be interpreted as an efficient approximation (based on iterating a contraction map) of the Proximal Point method with entropic regularizer when the feasibility set is the product of simplices.
\end{remark}
\begin{remark}
The idea of approximating the Proximal Point method through the iterations of a contraction map first appeared in the seminal work of Nemirovski establishing the $\mathcal{O}(1/T)$ convergence rate of the famous \textit{Mirror-Prox method} \cite{N04}. Nemirovksi mentions that \textit{"the prox-method becomes implementable: for
all computational purposes, its step requires a small number of fixed point iterations"} \cite{N04} and then he proves for the case of sets with bounded diameter only $2$ fixed-point iterations are needed. 
Our work indicates that by formalizing the above argument simpler proofs can be derived for both for time-average and last-iterate convergence.   
\end{remark}

\subsubsection*{Our Results}
Up next we summarize the convergence results for Clairvoyant Extra-Gradient and compare them with the results of existing for (EG) and (OGDA).
\begin{itemize}
    \item In Section~\ref{s:bounded} we establish that Clairvoyant Extra-Gradient admits $\Theta(\log T /T)$
    time-average convergence when $\mathcal{D}_1,\mathcal{D}_2$ admit bounded diameter. The $\Theta(1/T)$ time-average convergence was established in \cite{N04} for (EG) and in \cite{RS13} for (OGDA) (for bounded sets).

\item In Section~\ref{s:unbounded} we extend the $\Theta(\log T /T)$ convergence rate of Clairvoyant Extra-Gradient when $\mathcal{D}_1,\mathcal{D}_2$ admit unbounded diameter. The $\Theta(1/T)$ time-average convergence for (EG) was established in \cite{MS10}.

\item In Section~\ref{s:last-iterate} we establish the $\Theta(\log T /\sqrt{T})$ last-iterate convergence of Clairvoyant Extra-Gradient when $\mathcal{D}_1 = \mathbb{R}^n, \mathcal{D}_2 = \mathbb{R}^m$. Golowich et al. \cite{GPDO20} were the first to establish $\Theta(1/\sqrt{T})$ last-iterate convergence for (EG) when $\mathcal{D}_1 = \mathbb{R}^n, \mathcal{D}_2 = \mathbb{R}^m$.
\end{itemize}
\noindent Sections~\ref{s:bounded},~\ref{s:unbounded}~and~\ref{s:last-iterate} are self-contained and can be read independently.


\subsubsection*{Further Related Work}
The literature on the convergence properties of first-order methods in the convex/concave setting (or equivalently, monotone variational inequalities) is vast and thus impossible to review thoroughly in the context of this paper. In the following section, we present a small subset of recent works which will hopefully give the reader an idea of the overall picture of current research in this direction.

\textbf{Time-Average Convergence:} Nemirovski et al. \cite{N04} were the first to establish that (EG) admits $\Theta(1/T)$-time average convergence for bounded feasibility sets while Rakhlin et al. \cite{RS13} established the same results for (OGDA) \footnote{The above results refer to \textit{MirrorProx} and \textit{Optimistic Mirror Descent} which are the respective generalizations of (EG) and (OGDA).}. Monteiro et al. \cite{MS10} extended the convergence results for (EG) for the sets with unbounded diameter by using a different termination rule based on enlargement of the operator. 
For the special case $\mathcal{D}_1 = \mathbb{R}^n,\mathcal{D}_2 = \mathbb{R}^m$, Mohtari et al. \cite{MOP20} recently proposed a simpler convergence analysis for (EG)
by associating (EG) with (PP). Using the same approach, Mohtari et al. \cite{MOP20} were also able to establish $\Theta(1/T)$ time-average convergence for (ODGA) method when $\mathcal{D}_1 = \mathbb{R}^n,\mathcal{D}_2 = \mathbb{R}^m$. Moreover Antonakopoulos et al. \cite{AP21} study the convergence properties of  various first-order methods in the presence of noise while Antonakopoulos et al. \cite{AB21} provide adaptive first-order methods with optimal rates. Other recent works establishing convergence properties of (EG) or (OGDA) or variants thereof include \cite{GBVVL19,NO09,DP18,LS19,DISZ18,FMPV22}.

\textbf{Last-Iterate Convergence:} Daskalakis et al. \cite{DP18} and Liang et al. \cite{LS19} established that (OGDA) admits last-iterate convergence in the unconstrained case. Later, Daskalakis et al. \cite{DP19} extended the last-iterate results for \text{Optimistic MWU} for products of simplices. Golowich et al. \cite{GPDO20}  were the first to establish $\Theta(1/\sqrt{T})$ last-iterate convergence for (EG) when $\mathcal{D}_1 = \mathbb{R}^n,\mathcal{D}_2 = \mathbb{R}^m$, and indeed they proved that this rate is tight for a general class of first-order methods. Later, Gorbunov et al. \cite{GLG22} extended the last-iterate results of Golowich et al. \cite{GPDO20} by removing the assumption on the  Lipschitzness of the Jacobian of the operator. Very recently, Cai et al. \cite{COZ22} established $\Theta(1/\sqrt{T})$ last-iterate convergence for both (EG) and (ODGA) for general convex sets.

Going beyond the convex/concave setting, Diakonikolas et al. \cite{DDJ21} and Pethick et al. \cite{P22} established convergence for variations of (EG) for \textit{weak Minty variational inequalities}, while Mertikopoulos et al. \cite{MLZ19} established convergence properties for \textit{coherent} saddle-point problems. Finally the idea of efficiently approximating the Proximal Point method has been used in the context of Halpern iterations \cite{D20} and the catalyst framework \cite{YZKH20}.

\section{Preliminaries}
We denote with $\norm{z} \in \mathbb{R}^n$ the Euclidean norm of $z$, i.e. $\norm{z} = \sqrt{\sum_{i=1}^n z_i^2}$ and with $\left[ z \right]_{\mathcal{D}}$ the Euclidean projection of $z$ to the set $\mathcal{D}$, i.e. $\left[ z \right]_{\mathcal{D}} := \text{argmin}_{z' \in \mathcal{D}}\norm{z - z'}$. Moreover we denote with $\mathcal{B}(z,\rho)$ the ball of radius $\rho >0$ centered at $z$, i.e. $\mathcal{B}(z,\rho)= \{z' \in \mathbb{R}^n:~ \norm{z - z'}\leq \rho\}$. We finally denote with $|\mathcal{D}|$ the diameter of set $\mathcal{D}$, i.e. $|\mathcal{D}|:= \max_{z,z' \in \mathcal{D}}\norm{z - z'}$. Up next, we formally present the assumptions on $f(x,y)$.

\begin{assumption}[Convex/Concave]
The function $f: \mathcal{D}_1 \times \mathcal{D}_2 \mapsto \mathbb{R}$ is convex with respect to $x \in \mathcal{D}_1$ and concave with respect to $y\in \mathcal{D}_2$ iff
for all $(x,y) \in \mathcal{D}_1 \times \mathcal{D}_2$,
\begin{itemize}
    \item $f(x',y) \geq f(x,y) + \langle \nabla_x f(x,y), x' -x \rangle~~~~$ for all $x' \in \mathcal{D}_1$
    
    \item $f(x,y') \leq f(x,y) + \langle \nabla_y f(x,y), y' -y \rangle~~~~$ for all $y' \in \mathcal{D}_2$
\end{itemize}
\end{assumption}

\begin{assumption}[Smoothness]
The function $f: \mathcal{D}_1 \times \mathcal{D}_2 \mapsto \mathbb{R}$ is $L$-smooth with respect to $x \in \mathcal{D}_1$ and $L$-smooth with respect to $y\in \mathcal{D}_2$ if
for all $(x,y) \in \mathcal{D}_1 \times \mathcal{D}_2$,
\begin{itemize}
        \item $\norm{ \nabla _x f(x , y) - \nabla _x f(x' , y)} \leq L  \norm{x - x'}~~~~~~$     for all $x' \in \mathcal{D}_1$
    
        \item $\norm{ \nabla _y f(x , y) - \nabla _y f(x , y')} \leq L  \norm{y - y'}~~~~~~$     for all $y' \in \mathcal{D}_2$ 
\end{itemize}
\end{assumption}

\noindent To simplify notation, we follow standard variational inequality notation. We denote $\mathcal{D}:=\mathcal{D}_1 \times \mathcal{D}_2$ and $z = (x,y) \in \mathcal{D}$ as the decision variable. 
\begin{definition}\label{d:operator}
Given a function $f:\mathcal{D}_1 \times \mathcal{D}_2 \mapsto \mathbb{R}$, consider the operator $F:\mathcal{D}\mapsto \mathcal{D}$ defined as
\begin{equation}\label{eq:operator}
F(z):= \left( \nabla_x f(x,y) , -\nabla_y f(x,y)\right)
\end{equation}
\end{definition}
\noindent Up next we reformulate Problem~\ref{eq:minmax} using the notation that we introduced above.
\begin{proposition}\label{prop:eq_PP}
A pair of points $z^\ast = (x^\ast,y^\ast)$ is a min-max solution for Problem~\ref{eq:minmax} if and only if 
\begin{equation}\label{eq:minmax_equivalent}
\left \langle F(z^\ast) , z^\ast -z \right \rangle \leq 0~~~\text{for all }z\in \mathcal{D}.   
\end{equation}
\end{proposition}
\noindent The equivalence between Problem~\ref{eq:minmax} and Problem~\ref{eq:minmax_equivalent} follows directly by Corollary~\ref{c:1}.
\begin{corollary}\label{c:1}
For any $z = (x, y)\in \mathcal{D}$ and $z^\ast = (x^\ast, y^\ast) \in \mathcal{D}$,
\[f(x^\ast , y) - f(x , y^\ast) \leq \left \langle F(z^\ast) , z^\ast -z \right \rangle\]
\end{corollary}
\begin{proof}
Notice that $\left \langle F(z^\ast) , z^\ast -z \right \rangle = \langle \nabla_x f(x^\ast,y^\ast), x^\ast - x\rangle + \langle \nabla_y f(x^\ast,y^\ast), y - y^\ast \rangle$. By the convexity of the function $f(\cdot,y^\ast)$ we get $f(x^\ast, y^\ast) -  f(x, y^\ast) \leq \langle \nabla_x f(x^\ast,y^\ast), x^\ast - x\rangle$ and by the concavity of the function $f(x^\ast,\cdot)$
we get
 $f(x^\ast, y)-f(x^\ast, y^\ast) \leq \langle \nabla_y f(x^\ast,y^\ast), y^\ast - y\rangle$. 
\end{proof}

\noindent We conclude the section with the monotonicity property of the operator $F$ that is proven in \cite{N04}. 
\begin{lemma}[\cite{N04}]\label{l:monotonicity}
The operator $F: \mathcal{D} \mapsto \mathcal{D}$ of Equation~\ref{eq:operator} is monotone. More precisely, for any $z,z' \in \mathcal{D}$
\[\left \langle F(z) - F(z') , z - z' \right \rangle \geq 0\]
\end{lemma}

\section{Implementing the Proximal Point Method with Contraction Maps}\label{s:contraction}
The goal of this section is to present how the update rule of \eqref{eq:PP} can be efficiently computed once the step-size $\gamma < 1/L$. Before doing so, we briefly illustrate the convergence properties of \eqref{eq:PP}. 

Using the notation introduced in the previous section, the update rule of \eqref{eq:PP} can be equivalently described as follows,
\begin{definition}[Proximal Point update rule]\label{d:PP}
Given $z \in \mathcal{D}$ and a step-size $\gamma >0$, a point $z' \in \mathcal{D}$ satisfies the \textit{proximal point operator}, $z' \in \mathrm{PP}^\gamma(z)$, if and only if 
\begin{equation}\label{eq:PP_operator}
z' = \left[z - \gamma  F\left(z'\right) \right]_{\mathcal{D}}  
\end{equation}
\end{definition}
\noindent As already mentioned the update rule of Proximal Point requires the solution of a \textit{fixed-point problem} (see Equation~\ref{eq:PP_operator}). Solving fixed-point problems not only is (in general) computationally intractable but also given $z\in \mathcal{D}$ and $\gamma >0$ there might be several $z' \in \mathcal{D}$ satisfying Equation~\ref{eq:PP_operator} (or none in case $\mathcal{D}$ is unbounded). As we shall soon see, both problems can be alleviated once $\gamma < 1/L$. To this end we can consider that Proximal Point method selects as $z_{t+1}$ any of the points $z'\in \mathrm{PP}^\gamma(z_t)$. Up next, we briefly present both the time-average and the last-iterate convergence properties of \eqref{eq:PP}.

\begin{theorem}[Folkore]\label{l:time_average_PP}
Consider the sequence $z_0,z_1,\ldots, z_T\in \mathcal{D}$ such that $z_{t+1} \in \mathrm{PP}^\gamma(z_t)$. Then,
\[\sum_{t=0}^{T-1} \langle F(z_{t+1}) , z_{t+1} - z \rangle \leq \frac{|\mathcal{D}|^2}{2 \gamma} \]
for all $z \in \mathcal{D}$. Equivalently the time-averaged pair of points $(\hat{x},\hat{y}):= \sum_{t=0}^{T-1}z_{t+1} / T$ satisfies,
\[f(\hat{x},y) - \frac{|\mathcal{D}|^2}{2\gamma T} \leq f(\hat{x},\hat{y}) \leq f(x,\hat{y}) + \frac{|\mathcal{D}|^2}{2\gamma T}\]
for all $x,y \in \mathcal{D}_1\times \mathcal{D}_2$.
\end{theorem}
\noindent In their recent work studying the \textit{last-iterate} convergence properties of \eqref{eq:EG}, Golowich et al. \cite{GPDO20} established \textit{last-iterate} convergence for \eqref{eq:PP} when 
$\mathcal{D}= \mathbb{R}^{n+m}$. Notice that in this case the update rule of \eqref{eq:PP} takes the form 
$z' = z - \gamma F(z')$ and that Problem~\ref{eq:minmax} asks for a $z^\ast \in \mathbb{R}^{n+m}$ with $F(z^\ast) =0$. 
\begin{theorem}[\cite{GPDO20}]\label{l:GD21}
Consider the sequence $z_0,z_1,\ldots, z_T\in \mathbb{R}^{n+m}$ such that $z_{t+1} = z_t - \gamma F(z_{t+1})$. In case there exists $z^\ast \in \mathbb{R}^{n+m}$ such that $\norm{F(z^\ast)} =0$ then,
\[\norm{F(z_{T})} \leq \frac{\norm{z_{0} - z^\ast}^2}{\gamma \sqrt{T}}.\]
\end{theorem}
\begin{remark}
The above convergence results hold no matter the selection of the step-size $\gamma$. This implies that Proximal Point can achieve arbitrarily fast convergence rates by selecting $\gamma \rightarrow \infty$. The latter is possible due to the fact that Proximal Point assumes the solution of the fixed-point problem of Equation~\ref{eq:PP_operator} that as we explained is generally intractable. On the positive side, both the proofs of Theorem~\ref{l:time_average_PP} and Theorem~\ref{l:GD21} are very simple and intuitive. We additionally remark that both Theorem~\ref{l:time_average_PP} and Theorem~\ref{l:GD21} are presented only for the sake of exposition and are not required in the convergence results of \textit{Clairvoyant Extra-Gradient} that are presented in Sections~\ref{s:bounded},~\ref{s:unbounded}~and~\ref{s:last-iterate}. 
\end{remark}
\subsection{Approximating the Proximal Point method with Contraction Maps}
In this section, we present how the update rule of Proximal Point (see Definition~\ref{d:PP}) can be efficiently computed through a contraction map when $\gamma < 1/L$.  
More precisely, in Algorithm~\ref{alg:Extra} we present an algorithm that given an input point $z_{t} \in \mathcal{D}$ and a step-size $\gamma >0$, it outputs a point $z_{t+1} \in \mathcal{D}$ such that 
\[z_{t+1} \simeq \left[z_t - \gamma \cdot  F(z_{t+1}) \right]_{\mathcal{D}}\]
\begin{algorithm}[H]
  \caption{Clairvoyant Extra-Gradient Update Rule}\label{alg:Extra}
  \begin{algorithmic}[1]
  \State \textbf{Input:} Step-size $\gamma >0$, $k \in \mathbb{N}$, $z\in \mathcal{D}$
  \smallskip
\State $w_0 \leftarrow z$
  \ForAll{$m = 1, \cdots, k$}
  \smallskip
  \State $w_m \leftarrow \left[z - \gamma F(w_{m-1}) \right]_{\mathcal{D}}$
 \EndFor
\State \textbf{Return} $w_k$
  \end{algorithmic}
\end{algorithm}
\noindent In Lemma~\ref{l:contraction} we establish the fact that the output of Algorithm~\ref{alg:Extra} is approximately the same as the output of the update rule of the Proximal Point method. We remark that Lemma~\ref{l:contraction} holds for any choice of $\gamma < 1/L$ but we choose to set $\gamma = 1/2L$ for the sake of simplicity. Similar contraction arguments can be found in \cite{N04,PSS21}.
\begin{lemma}\label{l:contraction}
For any $z\in \mathcal{D}$ and $\gamma = 1/2L$ the fixed-point equation $\hat{z} =  \left[z - \gamma \cdot F(\hat{z})\right]_{\mathcal{D}}$ admits a unique solution $\hat{z}$ denoted as $\mathrm{PP}^\gamma(z)$. Moreover,
\[\norm{w_k - \mathrm{PP}^\gamma(z)} \leq \frac{|\mathcal{D}|}{2^k}\]
where $w_k$ is the output of Algorithm~\ref{alg:Extra}.
\end{lemma}

\begin{proof}
Let $\gamma = 1/2L$ and recall that $\mathrm{PP}^\gamma(z) = \left[z - \gamma F(\mathrm{PP}^\gamma(z)) \right]_{\mathcal{D}}$.
\begin{eqnarray*}
\norm{ w_k -  \mathrm{PP}^\gamma(z)}&=& 
\norm{ \left[z - \gamma F(w_{k-1})\right]_{\mathcal{D}} -  \left[z-\gamma F(\mathrm{PP}^\gamma(z))\right]_{\mathcal{D}}}\\
&\leq& \gamma\norm{F(w_{k-1}) - F(\mathrm{PP}^\gamma(z))}\\
&\leq& \gamma L \norm{w_{k-1} - \mathrm{PP}^\gamma(z)}\\
&=& \frac{1}{2} \norm{w_{k-1} - \mathrm{PP}^\gamma(z)}
\leq \frac{1}{2^k}\norm{w_0 - \mathrm{PP}^\gamma(z)} \leq \frac{\mathcal{|D|}}{2^k}.
\end{eqnarray*}
\end{proof}
\noindent To this end, the update rule described in Algorithm~\ref{alg:Extra} together with Lemma~\ref{l:contraction} provides us with a first-order method which is a very straightforward approximation of the Proximal Point method. We call this first-order method \emph{Clairvoyant Extra-Gradient} and we present it in Algorithm~\ref{alg:Extra_Gradient}\footnote{According to whether $\mathcal{D}$ is bounded or unbounded, $k$ admits slightly different forms that we present in the respective sections.}.\\

\begin{algorithm}[!htb]
  \caption{Clairvoyant Extra-Gradient}\label{alg:Extra_Gradient}
  \begin{algorithmic}[1]
  \State \textbf{Input:} $z_0\in \mathcal{D}$
  \smallskip
    \State $\gamma \leftarrow 1/2L$.
\smallskip
  \ForAll{$t = 0, \cdots, T$}
  \smallskip
\State $w_0 \leftarrow z_t$
  \smallskip
  \ForAll{$m = 1, \cdots, k:= \Theta(\log(|\mathcal{D}| T) )$}
  \smallskip
  \State $w_m \leftarrow \left[z_t - \gamma F(w_{m-1}) \right]_{\mathcal{D}}$
 \smallskip
 \EndFor
 \smallskip
\State $z_{t+1} \leftarrow w_{k}$
 \smallskip
 \EndFor
  \end{algorithmic}
\end{algorithm}

\noindent It is rather intuitive to understand why Clairvoyant Extra-Gradient inherits mal Point method described in Theorem~\ref{l:time_average_PP}~and~\ref{l:GD21} (in particular, $\Theta(1/T)$ time-average and $\Theta(1/\sqrt{T})$ last-iterate convergence). The reason is that by selecting $k:= \Theta(\log(|\mathcal{D}| T))$ we ensure that at Step~$8$ of Algorithm~\ref{alg:Extra_Gradient}, 
\begin{equation}\label{eq:error_term}
\norm{z_{t+1} - \left[z_t - \gamma F(z_{t+1}) \right]_{\mathcal{D}}}\leq \frac{1}{\mathrm{poly}(T)}
\end{equation}
which means that over the $T$ iterations of Step~$3$, the trajectories of Proximal Point and Clairvoyant Extra-Gradient are almost identical. Obviously the simple proofs of Theorem~\ref{l:time_average_PP}~and~\ref{l:GD21} do not directly carry over since we need to account for the error terms of Equation~\ref{eq:error_term}. However, these error terms can be handled without significant complications. As a result, the analysis of Clairvoyant Extra-Gradient follows closely the proof-steps of Proximal Point method making its analysis significantly more structured than the respective analysis of Extra-Gradient and Optimistic GDA.
\smallskip

\section{Time-Average Convergence
for Bounded Sets}\label{s:bounded}
In this section we present the time-average convergence properties of Clairvoyant Extra-Gradient in case $\mathcal{D}$ is a bounded convex set. The latter is formalized in Theorem~\ref{t:bounded}.
\begin{theorem}\label{t:bounded}
Let $z_0,z_1,\ldots, z_T$ be the sequence produced by Clairvoyant Extra-Gradient (Algorithm~\ref{alg:Extra_Gradient}) for $\gamma = 1/2L$ and $k := \mathcal{O}\left(\log(\max(L,1)\cdot\max(\norm{F(x_0)},1)\cdot |\mathcal{D}| \cdot T \right))$. Then the time-averaged pair of points $(\hat{x},\hat{y}):= \sum_{t=0}^{T-1}z_{t+1} / T$ satisfies
\[f(x,\hat{y}) - \frac{L|\mathcal{D}|^2+1}{ T}  \leq f(\hat{x}, \hat{y}) \leq f(x,\hat{y}) + \frac{L|\mathcal{D}|^2+1}{T}\]
for all $x,y \in \mathcal{D}$.
\end{theorem}
\noindent The proof of Theorem~\ref{t:bounded} follows directly by combining Lemma~\ref{l:bounded} with Corollary~\ref{c:1}. 
\begin{lemma}\label{l:bounded}
Let $\mathcal{D}$ be a bounded convex set and $z_0,z_1,\ldots, z_T \in \mathcal{D}$ be the sequence produced by \textit{Clairvoyant Extra-Gradient} (Algorithm~\ref{alg:Extra_Gradient})  with $k := \mathcal{O}\left(\log(\max(L,1)\cdot\max(\norm{F(x_0)},1)\cdot |\mathcal{D}| \cdot T) \right)$ and $\gamma = 1/2L$. Then, 
\[\sum_{t=0}^{T-1} \langle F(z_{t+1}) , z_{t+1} - z \rangle   \leq L  \norm{z_0 - z}^2 + 1\]
for all $z \in \mathcal{D}$.
\end{lemma}
\begin{proof}
\noindent To simplify notation we set $p_{t+1}:= \mathrm{PP}^\gamma(z_t)$ meaning that $p_{t+1} = \left[z_t - \gamma F(p_{t+1}) \right]_{\mathcal{D}}$ and thus
\begin{equation}\label{eq:14}
\left\langle z_t - \gamma F(p_{t+1}) - p_{t+1}, z - p_{t+1}\right \rangle \leq 0~~~~\text{for all }z \in \mathcal{D}.    
\end{equation}
\smallskip
\noindent Applying in Equation~(\ref{eq:14}) the three point identity, $(a - b)^\top(b - c) = \norm{a-c}^2/2 - \norm{a-b}^2/2 - \norm{b-c}^2/2$ with $\alpha = z$, $b = p_{t+1}$ and $c = z_t$ we get that for any $z \in \mathcal{D}$,
\begin{eqnarray*}
\langle F(p_{t+1}) , p_{t+1} - z \rangle   &\leq& \frac{\norm{z_t - z}^2}{2\gamma} - \frac{\norm{p_{t+1} - z}^2}{2\gamma} - \frac{\norm{p_{t+1} - z_t}^2}{2\gamma}\\
&=& \frac{\norm{z_t - z}^2}{2\gamma} - \frac{\norm{(p_{t+1}-z_{t+1}) + (z_{t+1} - z)}^2}{2\gamma} - \frac{\norm{p_{t+1} - z_t}^2}{2\gamma}\\
&\leq& \frac{\norm{z_t - z}^2}{2\gamma} - \frac{\norm{z_{t+1} - z}^2}{2\gamma}  + 
\frac{\norm{z_{t+1} - z} \norm{p_{t+1}-z_{t+1}}}{\gamma}
- \frac{\norm{p_{t+1} - z_{t+1}}^2}{2\gamma}\\
&\leq& \frac{\norm{z_t - z}^2}{2\gamma} - \frac{\norm{z_{t+1} - z}^2}{2\gamma}  + 
\frac{|\mathcal{D}|^2}{\gamma 2^k}
\end{eqnarray*}
where the last-inequality comes from the fact that $\norm{p_{t+1} - z_{t+1}} \leq |\mathcal{D}|/2^k$ (see Lemma~\ref{l:contraction}). Then by a telescoping summation we get that,
\[\sum_{t=0}^{T-1} \langle F(p_{t+1}) , p_{t+1} - z \rangle \leq \frac{\norm{z_0 - z}^2}{2\gamma} + \frac{|\mathcal{D}|^2 T }{\gamma 2^k}.\]
Notice that to this end we have established the claim of Theorem~\ref{t:bounded} for the $p_t$-sequence. In order to establish the claim for the $z_t$-sequence we use the fact that $\norm{z_t - p_t} \leq |\mathcal{D}|/2^k$.
\begin{eqnarray*}
\langle F(z_{t+1}) , z_{t+1} - z \rangle &\leq&
\langle F(p_{t+1}) , p_{t+1} - z \rangle + \norm{F(z_{t+1}) - F(p_{t+1})}\norm{ z_{t+1} - z}\\
&+& \norm{F(p_{t+1})}\norm{z_{t+1} - p_{t+1}}\\
&\leq& \langle F(p_{t+1}) , p_{t+1} - z \rangle + L\norm{z_{t+1} - p_{t+1}} \norm{ z_{t+1} - z}\\
&+& \left(\norm{F(z_0)} + L\cdot \norm{p_{t+1} - z_0}\right)\cdot \norm{z_{t+1} - p_{t+1}}\\
&\leq& \langle F(p_{t+1}) , p_{t+1} - z \rangle + 
\frac{3\max(L,1) \max(\norm{F(z_0)},1) |\mathcal{D}|^2}{2^k}\\
&\leq& \frac{\norm{z_0 - z}^2}{2\gamma} + \frac{2L|\mathcal{D}|^2 T }{ 2^k} + \frac{3\max(L,1) \max(\norm{F(z_0),1)} |\mathcal{D}|^2}{2^k}
\end{eqnarray*}
The proof of Theorem~\ref{t:bounded} follows by
summing from $t=0$ to $T-1$ and selecting $k:= \log(5\max(L,1)\cdot\max(\norm{F(x_0)},1)\cdot|\mathcal{D}|^2\cdot T )$.
\end{proof}
\section{Time-Average Convergence 
for Unbounded Sets}\label{s:unbounded}
In this section, we extend the time-average convergence results of Clairvoyant Extra-Gradient for convex sets $\mathcal{D}$ with unbounded diameter. We remark that for unbounded sets $\mathcal{D}$ it is possible that there is no min-max solution to Problem~\ref{eq:minmax_equivalent} and thus we have to assume that such a $z^\ast \in \mathcal{D}$ exists. 
\begin{assumption}\label{as:existence}
There exists $z^\ast \in \mathcal{D}$ such that
$\left\langle F(z^\ast), z^\ast - z \right\rangle\leq 0~~\text{ for all }z \in \mathcal{D}$.
\end{assumption}
\noindent The convergence properties of Clairvoyant Extra-Gradient when $\mathcal{D}$ admits unbounded diameter are formally stated in Theorem~\ref{c:unbounded}.
\begin{theorem}\label{c:unbounded}
Let $z_0,z_1,\ldots, z_T$ be the sequence produced by Clairvoyant Extra-Gradient (Algorithm~\ref{alg:Extra_Gradient}) for $\gamma := 1/2L$ and $k:= \mathcal{O}\left(\log(\max(L,1)\cdot\max(\norm{F(x_0)},1)T)\right)$ and consider the time-average pair of points $(\hat{x},\hat{y}):= \sum_{t=0}^{T-1}z_{t+1} / T$. Then for any $z = (x,y) \in \mathcal{D}$
\[f(x,\hat{y}) - \frac{L\max(\norm{z_0 - z^\ast },\rho)^2}{T}  \leq f(\hat{x}, \hat{y}) \leq f(x,\hat{y}) + \frac{L\max(\norm{z_0 - z^\ast },\rho)^2}{T}\]
where $\rho = \norm{z - z_0}$.
\end{theorem}

\begin{remark}
We remark by considering the bounded set $\mathcal{D}' := \mathcal{D}\cap \mathcal{B}(z_0,\rho)$ (for some $\rho > 0$) and using the algorithm of the previous section does not include Theorem~\ref{c:unbounded} since then the guarantee would only hold for $z \in \mathcal{B}(z_0,\rho)$.
\end{remark}

\noindent We dedicate the rest of this section to the proof of Theorem~\ref{c:unbounded}. As in Section~\ref{s:bounded}, we set $p_{t+1} := \mathrm{PP}^{\gamma}(z_t)$ so as to simplify notation. Moreover using the exact same arguments as in Lemma~\ref{l:contraction} (apart from the last bounding step) one can establish that the $p_{t+1}$ (\textit{Proximal Point}) and $z_{t+1}$ (\textit{Clairvoyant Extra-Gradient}) iterates are close to each other. More precisely,
\begin{equation}\label{eq:32}
  \norm{p_{t+1} - z_{t+1}} \leq \frac{\norm{p_{t+1} - z_t}}{2^k}
\end{equation}
\noindent Using Equation~\ref{eq:32} we establish Lemma~\ref{t:unbounded} which is the main technical contribution of the section. Then, the proof of Theorem~\ref{t:unbounded} follows directly by combining Lemma~\ref{t:unbounded} with Corollary~\ref{c:1}.
\begin{lemma}\label{t:unbounded}
Let $z_0,z_1,\ldots, z_T$ be the sequence of points produced by Clairvoyant Extra-Gradient (Algorithm~\ref{alg:Extra_Gradient})  for $k := \Theta\left(\log\left(\max(L,1) \cdot T \cdot \max(\norm{F(z_0)},1)\right)\right )$ and $\gamma = 1/2L$. Then for all $z \in \mathcal{D}$, 
\[\sum_{t=0}^{T-1} \langle F(z_{t+1}) , z_{t+1} - z \rangle   \leq \mathcal{O}\left(L \max(\norm{z_0 - z}, \norm{z_0 - z^\ast})^2\right)\]
\end{lemma}
\begin{proof}
As in the proof of Theorem~\ref{t:bounded} we first prove the claim of Lemma~\ref{t:unbounded} for the $p_t$-sequence. We then extend the claim for the $z_t$-sequence by using the fact that $\norm{z_{t+1} - p_{t+1}} \leq \norm{z_{t} - p_{t+1}}/2^k$.\\
Recall that $p_{t+1} = \left[z_t - \gamma F(p_{+1}) \right]_{\mathcal{D}}$ and thus
\begin{equation}\label{eq:45}
\langle F(p_{t+1}) , p_{t+1} - z \rangle   \leq \frac{\norm{z_t - z}^2}{2\gamma} - \frac{\norm{p_{t+1} - z}^2}{2\gamma} - \frac{\norm{p_{t+1} - z_t}^2}{2\gamma}~~\text{for all }z\in \mathcal{D}.
\end{equation}
\noindent Thus,
\begin{eqnarray*}
\langle F(p_{t+1}) , p_{t+1} - z \rangle   &\leq& \frac{\norm{z_t - z}^2}{2\gamma} - \frac{\norm{p_{t+1} - z}^2}{2\gamma} - \frac{\norm{p_{t+1} - z_t}^2}{2\gamma}\\
&=& \frac{\norm{z_t - z}^2}{2\gamma} - \frac{\norm{(p_{t+1}-z_{t+1}) + (z_{t+1} - z)}^2}{2\gamma} - \frac{\norm{p_{t+1} - z_t}^2}{2\gamma}\\
&\leq& \frac{\norm{z_t - z}^2}{2\gamma} - \frac{\norm{z_{t+1} - z}^2}{2\gamma}  + 
\frac{\norm{z_{t+1} - z} \norm{p_{t+1}-z_{t+1}}}{\gamma}
- \frac{\norm{p_{t+1} - z_t}^2}{2\gamma}\\
&\leq& \frac{\norm{z_t - z}^2 - \norm{z_{t+1} - z}^2}{2\gamma} 
+\frac{\norm{z_{t+1} - z}^2}{2\gamma T^4} + \frac{T^4 \norm{p_{t+1}-z_{t+1}}^2}{2\gamma}
- \frac{\norm{p_{t+1} - z_t}^2}{2\gamma}\\
&\leq& \frac{\norm{z_t - z}^2 - \norm{z_{t+1} - z}^2}{2\gamma} 
+\frac{\norm{z_{t+1} - z}^2}{2\gamma \cdot T^4} + \frac{ T^4 \norm{p_{t+1}-z_{t}}^2}{2^{k+1}\gamma}
- \frac{\norm{p_{t+1} - z_t}^2}{2\gamma}\\
&\leq&
\frac{\norm{z_t - z}^2 - \norm{z_{t+1} - z}^2}{2\gamma}+\frac{\norm{z_{t+1} - z}^2}{2\gamma  T^4}
- \frac{\norm{p_{t+1} - z_t}^2}{4\gamma}
\end{eqnarray*}
where in the fourth inequality we used the fact that $\norm{p_{t+1} - z_{t+1}} \leq \norm{p_{t+1} - z_{t}}/2^k$ while the last inequality comes from the fact that $T^4/2^k \leq 1/2$ since $k = \Omega(\log T)$. Finally, by a telescopic summation over all $t$ we get that,
\begin{eqnarray*}
\sum_{t=0}^{T-1} \langle F(p_{t+1}) , p_{t+1} - z \rangle   &\leq& \frac{\norm{z_0 - z}^2 }{2\gamma} 
+\sum_{t=0}^{T-1}\frac{\norm{z_{t+1} - z}^2}{2\gamma \cdot T^4}
- \sum_{t=0}^{T-1} \frac{\norm{p_{t+1} - z_t}^2}{4\gamma}\\
&\leq& \frac{\norm{z_0 - z}^2 }{2\gamma}
+\frac{2T^2\cdot \sum_{t=0}^{T-1}\norm{z_{t+1}-z_t}^2 + 2T^2\norm{z_0 - z}^2}{2\gamma \cdot T^4}
- \sum_{t=0}^{T-1} \frac{\norm{p_{t+1} - z_t}^2}{4\gamma}\\
&\leq& \frac{\norm{z_0 - z}^2 }{2\gamma}
+\frac{2T^2 \left(2\sum_{t=0}^{T-1}\norm{p_{t+1}-z_t}\right)^2 + 2T^2\cdot \norm{z_0 - z}^2}{2\gamma \cdot T^4}
- \sum_{t=0}^{T-1} \frac{\norm{p_{t+1} - z_t}^2}{4\gamma}\\
&\leq& \frac{\norm{z_0 - z}^2 }{2\gamma}
+\frac{8T^3 \sum_{t=0}^{T-1}\norm{p_{t+1}-z_t}^2 + 2T^2\cdot \norm{z_0 - z}^2}{2\gamma \cdot T^4}
- \sum_{t=0}^{T-1} \frac{\norm{p_{t+1} - z_t}^2}{4\gamma}
\end{eqnarray*}
where the third inequality follows by the fact that $\sum_{t=0}^{T-1}\norm{z_{t+1} - z_t} \leq \sum_{t=0}^{T-1}\norm{z_{t+1} - p_{t+1}} + \sum_{t=0}^{T-1}\norm{p_{t+1} - z_{t}}\leq  2 \sum_{t=0}^{T-1}\norm{p_{t+1} - z_t}$ (since $\norm{p_{t+1} - z_{t+1}}\leq \norm{p_{t+1} - z_{t}}/2^k$). As a result, for $T \geq 32$ we get that
\begin{equation}\label{eq:10}
\sum_{t=0}^{T-1} \langle F(p_{t+1}) , p_{t+1} - z \rangle
\leq \frac{\norm{z_0 - z}^2 }{\gamma}
- \sum_{t=0}^{T-1} \frac{\norm{p_{t+1} - z_t}^2}{8\gamma}.
\end{equation}
\noindent To this end we have established the bound for the $p_t$-sequence, but in order to complete the proof of Theorem~\ref{t:unbounded} we need a precise bound on the error term $\norm{z_t - p_t}$. To do so, we apply Equation~\ref{eq:10} for $z = z^\ast$ i.e. $\langle F(z^\ast) , z^\ast - z\rangle \leq 0$ for any $z \in \mathcal{D}$.
\begin{equation}\label{eq:11}
\sum_{t=0}^{T-1}\norm{p_{t+1} - z_t}^2 \leq 8 \norm{z_0 - z}^2 -8\gamma \underbrace{\sum_{t=0}^{T-1} \langle F(p_{t+1}) , p_{t+1} - z^\ast \rangle}_{\geq 0}
\end{equation}
where $\langle F(p_{t+1}) , p_{t+1} - z^\ast \rangle\geq 0$ comes from the fact that $\langle F(p_{t+1}) -F(z^\ast) , p_{t+1} - z^\ast \rangle\geq 0$ (Corollary~\ref{c:1}) and $\langle F(z^\ast) , z^\ast - p_{t+1}\rangle \leq 0$. As a result, we get the following explicit upper bound on the error term $\norm{z_{t
+1} - p_{t+1}}$, 
\begin{equation}\label{eq:error}
\norm{z_{t+1} - p_{t+1}} \leq \frac{\norm{z_t - p_{t+1}}}{2^k}\leq 2\sqrt{2}\cdot \frac{\norm{z_0 - z^\ast}}{2^k}.
\end{equation}
\noindent Using Equation~\ref{eq:32} and~\ref{eq:error} one can additionally establish that $\norm{z_{t+1} - z^\ast},\norm{p_{t+1} - z_0} \leq\mathcal{O}\left( T \norm{z_0 - z^\ast}\right)$ (see Corollary~\ref{c:2}). We conclude with the proof of Lemma~\ref{t:unbounded}: 
\begin{eqnarray*}
\sum_{t=0}^{T-1}\langle F(z_{t+1}) , z_{t+1} - z^\ast \rangle &\leq&
\sum_{t=0}^{T-1} \langle F(p_{t+1}) , p_{t+1} - z^\ast \rangle + \sum_{t=0}^{T-1}\norm{F(z_{t+1}) - F(p_{t+1})}\cdot \norm{ z_{t+1} - z^\ast}\\
&+& \sum_{t=0}^{T-1} \norm{F(p_{t+1})}\norm{z_{t+1} - p_{t+1}}\\
&\leq& \sum_{t=0}^{T-1} \langle F(p_{t+1}) , p_{t+1} - z^\ast \rangle + \mathcal{O}\left( \frac{LT^2 \norm{z_0 - z^\ast}^2}{2^k} \right)\\
&+& \mathcal{O}\left(\frac{T \norm{F(z_0)} \norm{z_0 - z^\ast}}{2^k}\right) + \mathcal{O}\left( \frac{LT^2 \norm{z_0 - z^\ast}^2}{2^k} \right)\\
&\leq& \frac{\norm{z_0 - z^\ast}^2}{\gamma} + \mathcal{O}\left( \frac{\max(L,1)T^2 \max(\norm{F(z_0)},1) \cdot \norm{z_0 - z^\ast}^2}{2^k} \right)\\
\end{eqnarray*}
\begin{corollary}\label{c:2}
The following inequalities hold,
\begin{enumerate}
    \item $\norm{z_{t+1} - z^\ast} \leq \sqrt{66} T \cdot \norm{z_0 - z^\ast}$
        \item $\norm{p_{t+1} - z_0} \leq \sqrt{66} T\cdot \norm{z_0 - z^\ast}$
\end{enumerate}
\end{corollary}
\begin{proof}
$\norm{z_{t+1} - z^\ast}^2 \leq 2T \sum_{t=0}^{T-1}\norm{z_{t+1}-z_t}^2 + 2T \norm{z_0 - z^\ast}^2 \leq 8T^2 \sum_{t=0}^{T-1}\norm{p_{t+1}- z_t}^2 + 2T \norm{z_0 -z^\ast}^2$ where the last inequality comes from $\sum_{t=0}^{T-1}\norm{z_{t+1} - z_t} \leq 2\sum_{t=0}^{T-1}\norm{p_{t+1} - z_t}$. As a result,
\[\norm{z_{t+1} - z^\ast}^2 \leq  8T^2 \sum_{t=0}^{T-1}\norm{p_{t+1}- z_t}^2 + 2T \norm{z_0 -z^\ast}^2\leq 66T^2\cdot \norm{z_0 -z^\ast}^2 \]
where the last inequality comes from the Equation~\ref{eq:error}. The proof of the second item follows by the exact same steps.
\end{proof}
\end{proof}

\section{Last-Iterate Convergence 
}\label{s:last-iterate}
In this section, we establish the last-iterate convergence properties of \textit{Clairvoyant Extra-Gradient} when $\mathcal{D}= \mathbb{R}^{n+m}$. As already mentioned, our results recover the recent result of \cite{GPDO20} establishing $\Theta(1/\sqrt{T})$ last-iterate convergence of Extra-Gradient when $\mathcal{D}= \mathbb{R}^{n+m}$. We remark that in this case, Problem~\ref{eq:minmax_equivalent} asks for a $z^\ast \in \mathbb{R}^{n+m}$ such that $\norm{F(z^\ast)} =0$ and thus we need to assume that such a point exists.
\begin{assumption}
There exists $z^\ast \in \mathbb{R}^{n+m}$ such that $\norm{F(z^\ast)} = 0$.
\end{assumption}
\noindent The last-iterate convergence properties of \textit{Clairvoyant Extra-Gradient} are formally stated and established in Theorem~\ref{t:last-iterate}.
\begin{theorem}\label{t:last-iterate}
Let $z_0,z_1,\ldots,z_T$ be the sequence of points produced by the Clairvoyant Extra-Gradient for $\mathcal{D} = \mathbb{R}^{n+m}$ and $k = \mathcal{O}\left(\log\left(\max(L,1) \cdot T \cdot \max(\norm{F(z_0)},1)\right)\right)$. Then
\[ \norm{F(z_T )} \leq \mathcal{O}\left(\frac{L \norm{z_0 - z^\ast}^2}{\sqrt{T}} \right).\]
\end{theorem}
\noindent We dedicate the rest of the section to the proof of Theorem~\ref{t:last-iterate}. As before we set $p_{t+1}:= \mathrm{PP}^\gamma(z_t)$ where $\gamma = 1/2L$. Since we are in the unconstrained case, the latter implies that $p_{t+1} = z_t - \gamma F(p_{t+1})$. \noindent Recall that by Equation~\ref{eq:11} derived in Section~\ref{s:unbounded} we know that,
\begin{equation}\label{eq:1}
\sum_{t=0}^{T-1} \norm{p_{t+1} - z_t}^2 \leq 8 \norm{z_0 - z^\ast}^2
\end{equation}
In order to prove Theorem~\ref{t:last-iterate}, we just need to show that the quantity $\norm{p_{t+1} - z_t}^2$ is decreasing. This statement is not entirely true, but the following lemma establishes that $\norm{p_{T} - z_{T-1}}$ is approximately the minimum over all iterates.
\begin{lemma}\label{l:31}
For any $t\geq 0$, $$\norm{p_{T} - z_{T-1}}^2 \leq  \norm{p_{t+1} - z_t}^2 + \mathcal{O}(T\norm{z_0 - z^\ast}^2/2^k).$$
\end{lemma}
\noindent Using Lemma~\ref{l:31} one can easily establish Theorem~\ref{t:last-iterate}. More precisely, combining Lemma~\ref{l:31} with Equation~\ref{eq:1} we directly get that
\[T \norm{p_{T} - z_{T-1} }^2 \leq \sum_{t=0}^{T-1} \norm{p_{t+1} - z_t}^2+ \mathcal{O}\left(T^2\frac{\norm{z_0 -z^\ast}^2}{2^k}\right) \leq \mathcal{O}\left(\norm{z_0 - z^\ast}^2\right)\]
As a result, $\norm{F(p_T)} \leq \mathcal{O}\left(\frac{\norm{z_0 - z^\ast}}{\gamma \sqrt{T}}\right)$. Also notice that $\norm{z_T - p_T} \leq \norm{z_{T-1} - p_T}/2^k \leq \mathcal{O}\left(\norm{z_0 - z^\ast}/2^k\sqrt{T} \right)$ and thus $\norm{F(z_T)} \leq \norm{F(p_T)} + L\norm{z_T - p_T} \leq \mathcal{O}\left(L\frac{\norm{z_0 - z^\ast}}{ \sqrt{T}}\right)$.

\noindent We conclude the section with the proof of Lemma~\ref{l:31}.
\begin{proof}[Proof of Lemma~\ref{l:31}]
We first establish that $\norm{p_{t+1} - p_t } \leq \norm{z_t - z_{t-1}}$. Notice that once this is established, Lemma~\ref{l:31} intuitively follows since $z_t \simeq p_t$.
\begin{eqnarray*}
\norm{p_{t+1} - p_t }^2 &=& \norm{z_t -\gamma F(p_{t+1}) - z_{t-1} +\gamma F(p_{t})}^2\\
&=& \norm{z_{t-1} - z_t }^2 + 2\gamma (F(p_{t}) - F(p_{t+1}))^\top (z_t - z_{t-1}) + \gamma^2 \norm{ F(p_{t}) - F(p_{t+1})}^2\\
&\leq& \norm{z_{t-1} - z_t }^2 + 2\gamma (F(p_{t}) - F(p_{t+1}))^\top (z_t - \gamma F(p_{t+1}) - z_{t-1}+\gamma F(p_{t+1}))\\
&=& \norm{z_{t-1} - z_t }^2 + 2\gamma (F(p_{t}) - F(p_{t+1}))^\top (p_{t+1} - p_t) \\
&\leq& \norm{z_{t-1} - z_t }^2
\end{eqnarray*}
\noindent Equation~\ref{eq:1} implies that $\norm{p_{t+1} -z_t} \leq 2\sqrt{2}\norm{z_0 - z^\ast}$ and thus by Equation~\ref{eq:45} we get that $\norm{z_{t+1} - p_{t+1}} \leq 2\sqrt{2}\norm{z_0 - z^\ast}/2^k$. Finally by the triangle inequality we get that $\norm{p_{t+1} - p_t},\norm{z_{t+1} - z_t} \leq 4\sqrt{2}\norm{z_0 - z^\ast}$. Having established the above bounds we conclude with the proof of Lemma~\ref{l:31}.
\begin{eqnarray*}
\norm{p_{t+1} - z_t}^2 &\leq& \norm{p_{t+1} - p_t}^2 + \norm{p_{t} - z_t}^2 + 2\norm{p_t - z_t}\norm{p_t - p_{t+1}}\\
&\leq& \norm{z_{t} - z_{t-1}}^2 + 72\norm{z_0 - z^\ast}^2/2^k\\
&\leq& \norm{p_{t} - z_{t-1}}^2 + \norm{z_t - p_t}^2 + 2\norm{z_t - p_t}\norm{p_t - z_{t-1}} + 72\norm{z_0 - z^\ast}^2/2^k\\
&=& \norm{p_{t} - z_{t-1}}^2 + 144\norm{z_0 - z^\ast}^2/2^k
\end{eqnarray*}
As a result, $\norm{p_{T} - z_{T-1}}^2 \leq  \norm{p_{t+1} - z_t}^2 + 144T\cdot \norm{z_0 - z^\ast}^2/2^k$.
\end{proof}


\section*{Acknowledgments}
This project has received funding from the European Research Council (ERC) under the European Union's Horizon 2020 research and innovation program (grant agreement n° $725594$), the Swiss National Science Foundation (SNSF) under grant number $200021\_205011$ and Innosuisse.

\bibliography{refs}
\bibliographystyle{plain}

\appendix

\section{Omitted Proofs of Section~\ref{s:contraction}}\label{s:average_iterate_PP}
\noindent In this section we present the proofs of Theorem~\ref{l:time_average_PP} and Theorem~\ref{l:GD21} characterizing the convergence properties of the Proximal Point method.

\begin{reptheorem}{l:time_average_PP}
Consider the sequence $z_0,z_1,\ldots, z_T\in \mathcal{D}$ such that $z_{t+1} \in \mathrm{PP}^\gamma(z_t)$. Then,
\[\sum_{t=0}^{T-1} \langle F(z_{t+1}) , z_{t+1} - z \rangle \leq \frac{|\mathcal{D}|^2}{2 \gamma} \]
for all $z \in \mathcal{D}$. Equivalently the time-averaged pair of points $(\hat{x},\hat{y}):= \sum_{t=0}^{T-1}z_{t+1} / T$ satisfies,
\[f(\hat{x},y) - \frac{|\mathcal{D}|^2}{2\gamma T} \leq f(\hat{x},\hat{y}) \leq f(x,\hat{y}) + \frac{|\mathcal{D}|^2}{2\gamma T}\]
for all $x,y \in \mathcal{D}_1\times \mathcal{D}_2$.
\end{reptheorem}
\begin{proof}
Since $z_{t+1} = \left[z_t - \gamma \cdot F(z_{t+1}) \right]_{\mathcal{D}}$ then
\[ \left\langle x_t - \gamma \cdot F(z_{t+1}) - z_{t+1}, z - z_{t+1}\right \rangle \leq 0~~~\text{for all }z \in \mathcal{D}\]
\noindent As a result, we get that
\begin{eqnarray*}
\sum_{t=0}^{T-1}\left[\langle F(z_{t+1}) , z_{t+1} - z \rangle\right]   &\leq& \sum_{t=0}^{T-1}
\frac{\langle z_t - z_{t+1} , z_{t+1} -z \rangle}{ \gamma}\\
&=& \sum_{t=0}^{T-1}\left[\frac{\norm{z_t - z}^2}{2\gamma} - \frac{\norm{z_{t+1} - z}^2}{2\gamma} - \frac{\norm{z_{t+1} - z_t}^2}{2\gamma}\right]\\
&\leq& \frac{\norm{z_0 - z}^2}{2 \gamma} - \sum_{t=0}^{T-1} \frac{\norm{z_{t+1} - z_t}^2}{2\gamma}
\end{eqnarray*}
where the second equality comes from the three-point identity
$\langle a - b , c - b\rangle = \frac{\norm{a-c}^2}{2} - \frac{\norm{a-b}^2}{2} - \frac{\norm{b-c}^2}{2}$.
The proof of Theorem~\ref{l:time_average_PP} directly follows by Lemma~\ref{l:time_average_PP} and Corollary~\ref{c:1}.
\end{proof}

\begin{reptheorem}{l:GD21}\cite{GPDO20}
Consider the sequence $z_0,z_1,\ldots, z_T\in \mathbb{R}^{n+m}$ such that $z_{t+1} = z_t - \gamma F(z_{t+1})$. In case there exists $z^\ast \in \mathbb{R}^{n+m}$ such that $\norm{F(z^\ast)} =0$ then,
\[\norm{F(z_{T})} \leq \frac{\norm{z_{0} - z^\ast}^2}{\gamma \sqrt{T}}.\]
\end{reptheorem}
\noindent The proof of Theorem~\ref{l:GD212} follows easily by the Lemma~\ref{l:GD21} established in \cite{GPDO20} showing that the distance covered by the Proximal Point method strictly decreases.
\begin{lemma}\label{l:GD212}
Let $z_{t+1} = z_t - \gamma F(z_{t+1})$ and $z_{t} = z_{t-1} - \gamma F(z_t)$. Then,
\[\norm{z_{t+1} - z_t} \leq \norm{z_{t} - z_{t-1}}.\]
\end{lemma}
\begin{proof}
\begin{eqnarray*}
\norm{z_{t+1} - z_t }^2 &=& \norm{z_t -\gamma F(z_{t+1}) - z_{t-1} +\gamma F(z_{t})}^2\\
&=& \norm{z_{t-1} - z_t }^2 + 2\gamma (F(z_{t}) - F(z_{t+1}))^\top (z_t - z_{t-1}) + \gamma^2 \norm{ F(z_{t}) - F(z_{t+1})}^2\\
&\leq& \norm{z_{t-1} - z_t }^2 + 2\gamma (F(z_{t}) - F(z_{t+1}))^\top (z_t - \gamma F(z_{t+1}) - z_{t-1}+\gamma F(z_{t}))\\
&=& \norm{z_{t-1} - z_t }^2 + 2\gamma (F(z_{t}) - F(z_{t+1}))^\top (z_{t+1} - z_t) \leq \norm{z_{t-1} - z_t }^2\\
\end{eqnarray*}
where the last inequality follows by Lemma~\ref{l:monotonicity}.
\end{proof}
\noindent We conclude the section with the proof of Theorem~\ref{t:last-iterate}. Applying Lemma~\ref{l:time_average_PP} for $z = z^\ast$ we get that 
\[\sum_{t=0}^{T-1} \frac{\norm{z_{t+1} - z_t}^2}{2\gamma} \leq  \frac{\norm{z_0 - z}^2}{2 \gamma} + \sum_{t=0}^{T-1} \underbrace{\langle F(z_{t+1} , z^\ast - z_{t+1}\rangle}_{\leq 0}\] 
Combining the fact that $\langle F(z_{t+1}) - F(z^\ast), z_{t+1} - z^\ast ) \geq 0$ (monotonicity of $F(\cdot)$ in Lemma~\ref{l:monotonicity}) with $F(z^\ast) =0$ we get that $\langle F(z_{t+1} , z^\ast - z_{t+1}\rangle \leq 0$. As a result,
\[ \gamma^2T \cdot \norm{F(z_{T})}^2 = T \cdot \norm{z_{T} - z_{T-1}}^2 \leq \norm{z_0 - z^\ast}^2.\]

\end{document}